\documentclass[english]{article}
\usepackage[T1]{fontenc}
\usepackage[latin9]{inputenc}
\usepackage{dsfont}
\usepackage{babel}
\usepackage{hyperref}

\usepackage[letterpaper]{geometry}
\usepackage{color}
\usepackage{amsthm}
\usepackage{amsmath}
\usepackage{amssymb}

\def\<{\langle}
\def\>{\rangle}

\def\dg{\dagger}
\def\tr{\mathrm{Tr}}
\def\e{\text{e}}
\def\M{\mathcal{M}}
\def\im{\imath}
\def\C{\mathbb{C}}

\def\id{\mathcal{I}}
\def\one{{\mathchoice {\mathrm{1\mskip-4mu l}} {\mathrm{1\mskip-4mu l}} %
{\mathrm{1\mskip-4.5mu l}} {\mathrm{1\mskip-5mu l}}}}

\def\red#1#2{\mathcal{R}_{#1}^{#2}}

\def\ket#1{|#1\rangle}

\def\proj#1#2{|#1\rangle\langle#2|}
\def\inner#1#2{\langle#1|#2\rangle}
\def\m#1{\mathbb{M}_{#1}}
\def\ew#1{\mathcal{W}_{#1}}

\newcommand{\ot}{{\,\otimes\,}}
\newcommand{{\spa}}{{structural physical approximation\ }}
\newcommand{{\cp}}{{completely positive\ }}

\newtheorem{thm}{Theorem}
\newtheorem{defn}{Definition}
\newtheorem{example}{Example}
\newtheorem{lem}{Lemma}
\newtheorem{prop}{Proposition}

\begin{document}

\date{}
\title{\textbf{New tools for investigating \\ positive maps in matrix algebras}}

\author{Justyna Pytel Zwolak$^1$ and Dariusz  Chru\'sci\'nski$^2$ \\
$^1$ Department of Physics, Oregon State University,\\
Corvallis, OR 97331, USA \\
$^2$ Institute of Physics, Nicolaus Copernicus University,\\
Grudzi\c{a}dzka 5/7, 87--100 Toru\'n, Poland}

\maketitle

\begin{abstract}
We provide a novel tool which may be used to construct new examples of positive maps in matrix algebras (or, equivalently, entanglement witnesses). 
It turns out that this can be used to prove positivity of several well known maps (such as reduction map, generalized reduction, Robertson map, and many others). Furthermore, we use it to construct a new family of linear maps and prove that they are positive, indecomposable and (nd)optimal.
\end{abstract}

\section{Introduction}

Entanglement is one of the essential features of quantum physics and is fundamental to future quantum technologies. Therefore, there is a tremendous interest in developing efficient theoretical and experimental methods to detect entanglement. Linear positive maps in matrix algebras \cite{Paulsen,Bhatia} provide a basic tool to discriminate between separable and entangled states of composed quantum systems \cite{HHHH,Guhne}. A quantum state represented by the density
operator $\rho$ in $\mathcal{S}(\mathcal{H}_A \ot \mathcal{H}_B)$ is separable if and
only if it can be represented as
\begin{equation}\label{}
    \rho = \sum_\alpha p_\alpha \rho^{(A)}_\alpha \ot
    \rho^{(B)}_\alpha\ ,
\end{equation}
where $p_\alpha$ denotes a probability distribution, and $\rho^{(A)}_\alpha$ and $\rho^{(B)}_\alpha$ are density operators of
subsystems A and B, respectively. It is clear that separable states
define a convex subset in the space of all density operators in
$\mathcal{S}(\mathcal{H}_A \ot \mathcal{H}_B)$. States which are not separable
are called entangled. It is well known that $\rho$ represents a separable state if and only if \cite{EW1}
\begin{equation}\label{}
    (\id_A \ot \Lambda)\rho \geq 0 \ ,
\end{equation}
for all linear positive maps $\Lambda : \mathcal{B}(\mathcal{H}_B) \rightarrow \mathcal{B}(\mathcal{H}_A)$, where $\id_A : \mathcal{B}(\mathcal{H}_A) \rightarrow \mathcal{B}(\mathcal{H}_A)$ denotes an identity map, i.e., $\id_A(X) = X$ for each $X \in \mathcal{B}(\mathcal{H}_A)$  and $\mathcal{B}(\mathcal{H})$ denotes a $\mathbb{C}^*$-algebra of bounded operators in $\mathcal{H}$. Throughout the paper all Hilbert spaces are finite dimensional and hence $\mathcal{B}(\mathcal{H})$ may be treated as a matrix algebra $\m{N}(\mathbb{C})\equiv\m{N}$, where ${\rm dim}\, \mathcal{H} =N$.  Due to the well known duality \cite{Paulsen,CJ} between linear maps $\Phi : \mathcal{B}(\mathcal{H}_B) \rightarrow \mathcal{B}(\mathcal{H}_A)$ and linear operators in $\mathcal{H}_A \ot \mathcal{H}_B$ one may equivalently formulate the separability problem in terms of entanglement witnesses \cite{EW1,EW2}. A Hermitian operator $\ew{}$ defined on the tensor product
$\mathcal{H}_A \ot \mathcal{H}_B$ is called an entanglement witness if and only if:
1) $\mbox{Tr}(\ew{}\sigma_{\rm sep})\geq 0$ for all separable states
$\sigma_{\rm sep}$, and 2) there exists an entangled state $\rho$
such that $\mbox{Tr}(\ew{}\rho)<0$ (one says that $\rho$ is detected by
$\ew{}$).

Let $\Lambda : \mathcal{B}(\mathcal{H}_A) \rightarrow
\mathcal{B}(\mathcal{H}_B)$ be a positive linear map.  One calls $\Lambda$ $k$-positive if the map
\begin{equation}\label{}
    \id_k \ot \Lambda : \m{k} \ot \mathcal{B}(\mathcal{H}_A) \longrightarrow
\m{k} \ot \mathcal{B}(\mathcal{H}_B)
\end{equation}
is positive. In the above formula $\m{k}$ denotes a linear space of $k \times k$ complex matrices. A positive map which is $k$-positive for each $k$ is called completely positive
(CP). Actually, if  ${\rm dim}\mathcal{H}_A =d_A$ and ${\rm dim}\mathcal{H}_B
=d_B$ then $\Lambda$ is CP iff it is $d$-positive with $d = \min\{d_A,d_B\}$.
Denoting by $\mathcal{P}_k$ a convex cone of $k$-positive map one has a natural chain of inclusions
$$   \mathcal{P}_d \subset \mathcal{P}_{d-1} \subset \ldots \subset \mathcal{P}_2 \subset \mathcal{P}_1 \ , $$
where $\mathcal{P}_1$ is a convex cone of positive maps and $\mathcal{P}_d$ a convex cone of CP maps $(d=\min\{d_A,d_B\})$. Recall, that $\Lambda$ is {\em co--positive} iff $\Lambda \circ T$ is positive, where $T$ denotes transposition with respect to a fixed basis. Similarly, $\Lambda$ is $k$-copositive iff $\Lambda \circ T$ is $k$-positive. Denoting by $\mathcal{P}^k$ a convex cone of $k$-copositive map one has a dual chain of inclusions
$$   \mathcal{P}^d \subset \mathcal{P}^{d-1} \subset \ldots \subset \mathcal{P}^2 \subset \mathcal{P}^1 \ , $$
where $\mathcal{P}^1$ and $\mathcal{P}^d$ stands for copositive and completely copositive maps, respectively.

For the reader's convenience, let us recall some basic definitions that we shall use throughout the paper:

\begin{defn} A positive map $\Lambda$ is decomposable if $\Lambda \in \mathcal{P}_d \cup \mathcal{P}^d$, that is,
\begin{equation}\label{}
    \Lambda = \Lambda_1 + \Lambda_2 \circ T\ ,
\end{equation}
where $\Lambda_1$ and $\Lambda_2$ are CP and $T$ denotes
transposition in a given basis. Maps which are not decomposable are
called indecomposable (or nondecomposable).
\end{defn}


\begin{defn} A positive map $\Lambda$ is optimal if and only if for any completely positive map $\Phi_{CP}$, the map $\Lambda - \Phi_{CP}$ is no longer positive.
\end{defn}

\begin{defn} A positive map $\Lambda$ is nd-optimal if and only if for any decomposable map $\Phi_{D}$, the map $\Lambda - \Phi_{D}$ is no longer positive.
\end{defn}

\hspace{-0.55cm}Using the Jamio{\l}kowski isomorphism \cite{CJ} between operators and linear maps
one can extend the properties derived for positive maps to entanglement witnesses.

\begin{defn}
An entanglement witness $\ew{\Lambda}:=(\id_A\ot\Lambda)P^+_A$ is (in)decomposable, optimal and nd-optimal if and only if the corresponding positive map $\Lambda$ is (in)decomposable, optimal and nd-optimal.
\end{defn}

\hspace{-0.55cm}In particular one proves the two following interesting results concerning optimal and nd-optimal EW  \cite{Lew}.

\begin{thm} Let $\ew{}$ be an EW in $\mathcal{H}_A \ot \mathcal{H}_B$. If a set of  product vectors $\psi \ot \phi$ satisfying
\begin{equation}\label{}
    \< \psi \ot \phi |\ew{}|\psi\ot \phi\> = 0\ ,
\end{equation}
spans $\mathcal{H}_A \ot \mathcal{H}_B$, then $\ew{}$ is an optimal EW.
\end{thm}

\begin{thm}\label{thm:nd-optimal}
An entanglement witness $\ew{}$ is nd-optimal if and only if both $\ew{}$ and $\ew{}^\Gamma$ are optimal.
\end{thm}

Optimal positive maps (or, equivalently, optimal entanglement witnesses) provide the most efficient tool to discriminate between separable and entangled states. It is well known that any entangled state may be detected by some optimal map. In recent years there has been considerable effort in constructing and analyzing the structure of EWs
\cite{P1}--\cite{PO}. In this paper we provide a novel tool which may be used to construct new examples of positive maps (entanglement witnesses). It is based on a class of positive matrices discussed in the next section. We show that it may be used to prove positivity of several well known maps (reduction map, generalized reduction, Robertson map and many others). Further, we provide a new family of maps and prove that they are positive, indecomposable, and even both optimal and nd-optimal.

\section{A class of positive definite matrices}


In this section we provide a class of positive definite matrices that enables one to construct positive maps in matrix algebras. Let us start by recalling a well--known lemma.

\begin{lem}[\cite{Paulsen,Bhatia}]\label{thm:Bathia} A block matrix $M \in \mathbb{M}_{n+k}$
\begin{equation}\label{M}
    M= \left(\begin{array}{cc}
A & X\\
X^{\dg} & B
\end{array}\right)\ ,
\end{equation}
with $A\in \mathbb{M}_n$ and $B \in \mathbb{M}_k$ together with $A\ge0$ and $B>0$, is positive if and only if $A\ge XB^{-1}X^{\dg}$.
\end{lem}

\hspace{-0.55cm}We shall use this result to prove the following

\begin{thm}
\label{thm:twierdzenie} Let $\M_{N}^{K}$ be a matrix in $\m{K\cdot N} = \m{N} \ot \m{K} =: \m{N}(\m{K})$ of the following form:

\begin{equation}
\M_{N}^{K}=\left(\begin{array}{c|c|c|c}
(1-\alpha_{1})\one_{K} & -z_{12}M_{12} & \cdots & -z_{1N}M_{1N}\\
\hline -z^*_{12}M_{12}^{\dg} & (1-\alpha_{2})\one_{K} & \cdots & -z_{2N}M_{2N}\\
\hline \vdots & \vdots & \ddots & \vdots\\
\hline -z^*_{1N}M_{1N}^{\dg} & -z^*_{2N}M_{2N}^{\dg} & \cdots & (1-\alpha_{N})\one_{K}
\end{array}\right)
\end{equation}
with $\sum_{i=1}^{N}\alpha_{i}=1$ $(0\le\alpha_{i}\le1$ for $i=1,\dots,N)$,
$|z_{ij}|\le1$, and $M_{ij}\in\mathbb{M}_{K}(\mathbb{C})$, for $1\le i<j\le N$ such that

\begin{equation}\label{}
    M_{ij} M_{ij}^{\dg} = \alpha_j M_{ii}\ .
\end{equation}
If the blocks $M_{ij}$ of the matrix $\M_{N}^{K}$
satisfy the following properties:
\begin{eqnarray*}
\text{1.} & M_{ij}M_{kj}^{\dg} &= \,\, \alpha_{j}M_{ik},\\
\text{2.} & M_{ii} & \le \,\,  \alpha_{j}\one_{K},
\end{eqnarray*}
 then matrix $\M_{N}^{K}$ is positive--semidefinite.
\end{thm}

\begin{proof}
We will perform a proof by induction with respect
to the number of blocks $N$ in a matrix $\M_{N}^{K}$. Let
us assume that the $\M_{N-1}^{K}$ matrix is positive. From Theorem
\ref{thm:Bathia} we know that to prove positivity of matrix $\M_{N}^{K}$
it is enough to show that the following inequality holds

\begin{equation}\label{eq:inequality}
\M_{N-1}^{K} \ge \frac{\alpha_{N}}{1-\alpha_{N}}\mathcal{M}(\mathbf{z},M_{ij}),
\end{equation}
with

\begin{equation}
\mathcal{M}(\mathbf{z},M_{ij}):=\left(\begin{array}{c|c|c|c}
|z_{1N}|^{2}M_{11} & z_{1N}z^*_{2N}M_{12} & \cdots & z_{1N}z^*_{N-1,N}M_{1,N-1}\\
\hline z_{2N}z^*_{1N}M_{12}^{\dg} & |z_{2N}|^{2}M_{22} & \cdots & z_{2N}z^*_{N-1,N}M_{2,N-1}\\
\hline \vdots & \vdots & \ddots & \vdots\\
\hline z_{N-1,N}z^*_{1N}M_{1,N-1}^{\dg} & z_{N-1,N}z^*_{2N}M_{2,N-1}^{\dg} & \cdots & |z_{N-1,N}|^{2}M_{N-1,N-1}
\end{array}\right)\ge0,
\end{equation}
where the last inequality is a natural consequence of the construction. We introduce a normalization procedure for coefficients
$\alpha_{i}$ in a following way

\begin{equation}
\alpha'_{i}=\frac{\alpha_{i}}{1-\alpha_{N}}\quad,\text{ for }i=1,\dots,N-1,
\end{equation}
where $\sum_{i=1}^{N-1}\alpha'_{i}=1$. Applying this normalization
for submatrices $M_{ij}$ gives us

\begin{equation}
M'_{ij}=\sqrt{\frac{\alpha_{i}'\alpha'_{j}}{\alpha_{i}\alpha_{j}}}M_{ij}\quad,\text{ for }1\le i<j\le N-1.
\end{equation}
To show inequality (\ref{eq:inequality}) it is enough to prove that

\begin{equation}
\M_{\beta}\equiv\left(\begin{array}{c|c|c|c}
B_{1} & -z'_{12}M'_{12} & \cdots & -z'_{1,N-1}M'_{1,N-1}\\
\hline -z'^*_{12}(M'_{12})^{\dg} & B_{2} & \cdots & -z'_{2,N-1}M'_{2,N-1}\\
\hline \vdots & \vdots & \ddots & \vdots\\
\hline -z'^*_{1,N-1}(M'_{1,N-1})^{\dg} & -z'^*_{2,N-1}(M'_{2,N-1})^{\dg} & \cdots & B_{N}
\end{array}\right)\ge0,
\end{equation}
where
\begin{eqnarray*}
B_{i} & = & (1-\alpha'_{i}(1-\alpha_{N}))\one_{K}-|z_{iN}|^{2}\alpha_{N}M'_{ii},\\
z'_{ij} & = & (1-\alpha_{N})z_{ij}+\alpha_{N}z_{iN}z^*_{jN},
\end{eqnarray*}
with

\begin{equation}
|z'_{ij}|\le(1-\alpha_{N})|z_{ij}|+\alpha_{N}|z_{iN}z^*_{jN}|\le1.
\end{equation}
A simple calculation shows that

\begin{equation}
M'_{ij}(M'_{kj})^{\dg}=\frac{\alpha'_{j}}{\alpha_{j}}\sqrt{\frac{\alpha_{i}'\alpha'_{k}}{\alpha_{i}\alpha_{k}}}\, M_{ij}M_{kj}^{\dg}=\frac{\alpha'_{j}}{\alpha_{j}}\sqrt{\frac{\alpha_{i}'\alpha'_{k}}{\alpha_{i}\alpha_{k}}}\,\alpha_{j}M_{ik}=\alpha'_{j}\, M'_{ik}.
\end{equation}
By replacing $k$ with $i$ in the above formula one gets

\begin{equation}
M'_{ij}(M'_{ij})^{\dg}=\alpha'_{j}\, M'_{ii}
\end{equation}
and, due to the assumption,

\begin{equation}
M'_{ii}=\frac{\alpha_{i}'}{\alpha_{i}}M_{ii}\le\frac{\alpha_{i}'}{\alpha_{i}}\,\alpha_{i}\one_{K}=\alpha'_{i}\one_{K}.
\end{equation}
The last inequality implies

\begin{equation}
B_{i} \ge (1-\alpha'_{i}(1-\alpha_{N}))\one_{K}-|z_{iN}|^{2}\alpha_{N}\alpha'_{i}\one_{K} \ge  (1-\alpha'_{i})\one_{K}.
\end{equation}
As a consequence one finds

\begin{equation}
\M_{\beta}\ge\left(\begin{array}{c|c|c|c}
(1-\alpha'_{1})\one_{K} & -z'_{12}M'_{12} & \cdots & -z'_{1,N-1}M'_{1,N-1}\\
\hline -z'^*_{12}(M'_{12})^{\dg} & (1-\alpha'_{2})\one_{K} & \cdots & -z'_{2,N-1}M'_{2,N-1}\\
\hline \vdots & \vdots & \ddots & \vdots\\
\hline -z'^*_{1,N-1}(M'_{1,N-1})^{\dg} & -z'^*_{2,N-1}(M'_{2,N-1})^{\dg} & \cdots & (1-\alpha'_{N-1})\one_{K}
\end{array}\right)=\M_{N-1}^{K}.
\end{equation}
Since, by assumption, $\M_{N-1}^{K}$ is a positive matrix, we have completed the proof.
\end{proof}

\section{New proofs of positivity for a series of linear maps}

In this section we use Theorem \ref{thm:twierdzenie} to provide new proofs of positivity for a series of well--known maps. Let us recall that to prove positivity of a given map $\Lambda : \mathcal{B}(\mathcal{H}_A) \rightarrow \mathcal{B}(\mathcal{H}_B)$ it is enough to show that each rank-1 projector $P\in\mathcal{B}(\mathcal{H}_A)$ is mapped \textit{via} $\Lambda$ into a positive element  in $\mathcal{B}(\mathcal{H}_B)$.

\subsection{Generalized Reduction map}

Let us start our consideration with a generalized reduction map, $\red N{\mathbf{z}}:\m N\rightarrow\m N$, defined by

\begin{equation}
\red N{\mathbf{z}}(e_{ij})=\frac{1}{N-1}\begin{cases}
\one_{N}-e_{ii} & \text{for }i=j,\\
-z_{ij}e_{ij} & \text{for }i<j,
\end{cases}
\end{equation}
where $e_{ij}\in\m N$ stands for fixed orthonormal basis and $\mathbf{z}=\{z_{12,}z_{13,}\dots,z_{N-1,N}\}$
denotes a vector of complex numbers such that $|z_{ij}|\le1$.  Note, that if $z_{ij}=1$, then the above formula reproduces the standard normalized reduction map

\begin{equation}\label{}
    \mathcal{R}_N(X) = \frac{1}{N-1} \Big( \one_N {\rm Tr}(X) - X \Big) \ .
\end{equation}
Let us consider a rank-1 projector $P_{N}=\proj{\psi}{\psi}$, with $\psi=\bigoplus_{i=1}^{N}\sqrt{\alpha_{i}}x_{i}$, and $x_{i}\in\mathbb{C}$, $\alpha_{i}\in[0,1]$, $\sum_{i=1}^{N}\alpha_{i}=1$. Without loosing generality we can assume $|x_{i}|^{2}=1$ for all $i=1,\dots,N$. Now,

\begin{equation}
\red N{\mathbf{z}}(P_{N}) = \left[\begin{array}{cccc}
1-\alpha_{1} & -z_{12}M_{12} & \cdots & -z_{1N}M_{1N}\\
-z^*_{12}M_{12}^{\dg} & 1-\alpha_{2} & \cdots & -z_{2N}M_{2N}\\
\vdots & \vdots & \ddots & \vdots\\
-z^*_{1N}M_{1N}^{\dg} & -z^*_{2N}M_{2N}^{\dg} & \cdots & 1-\alpha_{N}
\end{array}\right]=\M_{N}^{1}\ ,
\end{equation}
with

\begin{equation}
M_{ij}=\sqrt{\alpha_{i}\alpha_{j}}x_{i}x^*_{j}\, ,\qquad\text{ for }1\le i<j\le N
\end{equation}
As we can see $|z_{ij}|=1$ by definition of $\red N{\mathbf{z}}$ and
for all off--diagonal blocks of the matrix $\M_{N}^{1}$ the following holds

\begin{equation}
M_{ij}M_{kj}^{\dg}=\left(\sqrt{\alpha_{i}\alpha_{j}}x_{i}x^*_{j}\right)\left(\sqrt{\alpha_{k}\alpha_{j}}x_{k}x^*_{j}\right)^{\dg}=\alpha_{j}\sqrt{\alpha_{i}\alpha_{k}}\,|x_{j}|^{2}x_{i}x^*_{k}=\alpha_{j}M_{ik}
\end{equation}
for all $1\le i<j\le N$, in particular $M_{ij}M_{ij}^{\dg}=\alpha_{j}M_{ii}$, and

\begin{equation}
M_{ii}=\sqrt{\alpha_{i}\alpha_{i}}x_{i}x^*_{i}\le\alpha_{i}\one_{1}.
\end{equation}
We have shown that conditions (1) and (2) of a Theorem \ref{thm:twierdzenie}
are satisfied for a matrix $\M_{N}^{1}$ which proves positivity of
a generalized Reduction map.

\subsection{Robertson map}

Let us now consider  an action of a well-known Robertson map

\begin{equation}
\Psi_{Rob}(X)=\frac{1}{2}\left(\begin{array}{c|c}
\one_{2}\tr X_{22} & -[X_{12}+\red 2{}(X_{21})]\\
\hline -[X_{21}+\red 2{}(X_{12})] & \one_{2}\tr X_{11}
\end{array}\right),
\end{equation}
where $X_{ij}\in\m 2$ and $\red 2{}$ stands
for a reduction map in $\m 2$, on a rank-1 projector $P_{4}=\proj{\psi}{\psi}$, with $\psi=\sqrt{\alpha_{1}}\psi_{1}\oplus\sqrt{\alpha_{2}}\psi_{2}$ ($\psi_{i}\in\mathbb{C}^{2}$, $\alpha_{i}\in[0,1]$ for $i=1,2$ and
$\alpha_{1}+\alpha_{2}=1$). Again without loosing generality we assume
$\inner{\psi_{i}}{\psi_{i}}=1$, $i=1,2$. Rewriting the reduction
map as $\red 2{}(X)=\sigma_{y}X^{T}\sigma_{y}^{\dg}$ allows us to represent $\Psi_{Rob}(P_4)$ as

\begin{equation}
\Psi_{Rob}(P_{4})=\frac{1}{2}\left(\begin{array}{c|c}
(1-\alpha_{1})\one_{2} & -M_{12}\\
\hline -M_{12}^{\dg} & (1-\alpha_{2})\one_{2}
\end{array}\right),
\end{equation}
with the off--diagonal blocks of the form

\begin{equation}
M_{12}=\sqrt{\alpha_{1}\alpha_{2}}\left[\proj{\psi_{1}}{\psi_{2}}+\sigma_{y}\proj{\psi^*_{2}}{\psi^*_{1}}\sigma_{y}^{\dg}\right]
\end{equation}
and $z_{12}=1$. We want to check whether conditions (1) and (2) of
the Theorem \ref{thm:twierdzenie} are satisfied. Since for any antisymmetric
and unitary matrix $U$ one has $\inner{\psi}{U\psi^*}=0$, thus
in particular for $\sigma_{y}$ one has $\inner{\psi_{1}}{\sigma_{y}\psi^*_{1}}=\inner{\psi_{2}}{\sigma_{y}\psi^*_{2}}=0$,
and as a consequence
\begin{equation}
M_{12}M_{12}^{\dg}  =  \alpha_{1}\alpha_{2}\left[\proj{\psi_{1}}{\psi_{1}}+\sigma_{y}\proj{\psi^*_{1}}{\psi^*_{1}}\sigma_{y}^{\dg}\right]=\alpha_{2}M_{11}.
\end{equation}
Now, because $\ket{\psi_{1}}$ and $\sigma_{y}\ket{\psi^*_{1}}$
are two normalized orthonormal vectors, they define an orthonormal
decomposition of an identity matrix and thus

\begin{equation}\label{eq:robertson1}
M_{ii}=\alpha_{1}\left[\proj{\psi_{1}}{\psi_{1}}+\sigma_{y}\proj{\psi^*_{1}}{\psi^*_{1}}\sigma_{y}^{\dg}\right]=\alpha_{i}\one_{2}
\end{equation}
which completes the proof.

\subsection{Generalization of the Robertson map \cite{Justyna2}}

Let us recall a generalization of the Robertson map to the  $\m{4N}$ algebra, given by

\begin{equation}
\Psi_{4N}(X)=\frac{1}{2N}\left(\begin{array}{c|c}
\one_{2N}\tr X_{22} & -[X_{12}+UX_{21}^{T}U^{\dg})]\\
\hline -[X_{21}+UX_{12}^{T}U^{\dg})] & \one_{2N}\tr X_{11}
\end{array}\right),
\end{equation}
with $U\in\m{2N}$ denoting an arbitrary antisymmetric and unitary matrix. Acting with a map $\Psi_{4N}$ on a projector $P_{4N}=\proj{\psi}{\psi}$,
with $\psi=\sqrt{\alpha_{1}}\psi_{1}\oplus\sqrt{\alpha_{2}}\psi_{2}$
($\psi_{i}\in\C^{2N}$, $\alpha_{1},\alpha_{2}\in[0,1]$, $\alpha_{1}+\alpha_{2}=1$
and $\inner{\psi_{i}}{\psi_{i}}=1$), leads to

\begin{equation}
\Psi_{4N}(P_{4N})=\frac{1}{2N}\left(\begin{array}{c|c}
(1-\alpha_{1})\one_{2N} & -M_{12}\\
\hline -M_{12}^{\dg} & (1-\alpha_{2})\one_{2N}
\end{array}\right),
\end{equation}
with $M_{ij}=\sqrt{\alpha_{i}\alpha_{j}}\left[\proj{\psi_{i}}{\psi_{j}}+U\left(\proj{\psi_{j}}{\psi_{i}}\right)^{T}U^{\dg}\right]$ and $z_{ij}=1$. Direct calculation shows that

\begin{equation}
M_{12}M_{12}^{\dg} = \alpha_{1}\alpha_{2}\left[\proj{\psi_{1}}{\psi_{1}}+U\proj{\psi^*_{1}}{\psi^*_{1}}U^{\dg}\right]=\alpha_{2}M_{11}
\end{equation}
and

\begin{equation}
M_{ii}=\alpha_{i}\left[\proj{\psi_{i}}{\psi_{i}}+U\left(\proj{\psi_{i}}{\psi_{i}}\right)^{T}U^{\dg}\right]\le\alpha_{i}\one_{2N}.
\end{equation}
This is what we sought out to be proved (the last inequality is a consequence of a simple fact that $\ket{\psi_{i}}$ and $U\ket{\psi_{i}}$ are two orthonormal
vectors and can be completed to a full orthonormal decompoition of
the identity).

\subsection{Complex extension of the Robertson map \cite{Justyna3}}

Both conditions from the Theorem \ref{thm:twierdzenie} are satisfied by the off--diagonal blocks of a  matrix obtained from acting on a rank-1 projector $P_{2N}$ with a map $\Psi_{2N}:\m{2N}\rightarrow\m{2N}$ defined as

\begin{equation}
\Psi_{2N}(X)=\frac{1}{2(N-1)}\left[\begin{array}{c|c|c|c}
A_{1} & -z_{12}B_{12} & \cdots & -z_{1N}B_{1N}\\
\hline -z^*_{12}B_{21} & A_{2} & \cdots & -z_{2N}B_{2N}\\
\hline \vdots & \vdots & \ddots & \vdots\\
\hline -z^*_{1N}B_{N1} & -z^*_{2N}B_{N2} & \cdots & A_{N}
\end{array}\right],
\end{equation}
with
\begin{eqnarray*}
A_{i} & = & \one_{2}\left(\tr X-\tr X_{ii}\right),\quad\text{ for }i=1,\dots,N\\
B_{ij} & = & X_{ij}+\red 2{}(X_{ji}),\quad\text{ for }1\le i<j\le N
\end{eqnarray*}
and $|z_{ij}|\le1$ for $1\le i<j\le 2N$, that is,

\begin{equation}
\Psi_{2N}(P_{2N})=\frac{1}{2(N-1)}\left[\begin{array}{c|c|c|c}
(1-\alpha_{1})\one_{2} & -z_{12}M_{12} & \cdots & -z_{1N}M_{1N}\\
\hline -z^*_{12}M_{12}^{\dg} & (1-\alpha_{2})\one_{2} & \cdots & -z_{2N}M_{2N}\\
\hline \vdots & \vdots & \ddots & \vdots\\
\hline -z^*_{1N}M_{1N}^{\dg} & -z^*_{2N}M_{2N}^{\dg} & \cdots & (1-\alpha_{N})\one_{2}
\end{array}\right],
\end{equation}
where $P_{2N}=\proj{\psi}{\psi}$ ($\psi=\bigoplus_{i=1}^{N}\sqrt{\alpha_{i}}\psi_{i}$
with $\psi_{i}\in\C^{2}$, $\sum_{i=1}^{N}\alpha_{i}=1$, $\alpha_{i}\in[0,1]$
and $\inner{\psi_{i}}{\psi_{i}}=1$ for $i=1,\dots,N$) and the off--diagonal blocks are defined as follows:

\begin{equation}
M_{ij}=\sqrt{\alpha_{i}\alpha_{j}}\left[\proj{\psi_{i}}{\psi_{j}}+\red 2{}\left(\proj{\psi_{j}}{\psi_{i}}\right)\right].
\end{equation}
Indeed, simple calculation leads to

\begin{equation}
M_{ij}M_{kj}^{\dg}  =  \alpha_{j}\sqrt{\alpha_{i}\alpha_{k}}\left[\proj{\psi_{i}}{\psi_{k}}+\sigma_{y}\proj{\psi^*_{i}}{\psi^*_{k}}\sigma_{y}^{\dg}\right]=\alpha_{j}M_{ik}.
\end{equation}
In particular, $M_{ij}M_{ij}^{\dg}=\alpha_{j}M_{ii}$. Also, analogously to Eq. (\ref{eq:robertson1}), $M_{ii}=\alpha_{i}\one_{2}$.

\section{A new class of maps in $\m{N}\ot \m{2K}$}

In this section we provide a new class of positive maps in $\m{N\cdot 2K}$.
Any matrix in $\m{N\cdot2K}$ may represented as a block $N\times N$ matrix in $\m{N}(\m{2K})$. Let us define a map $\Psi:\m{N\cdot2K}\rightarrow\m{N\cdot2K}$
in the following way
\[
\Psi(X)=\frac{1}{2K(N-1)}\left[\begin{array}{c|c|c|c}
A_{1} & -z_{12}B_{12} & \cdots & -z_{1N}B_{1N}\\
\hline -z^*_{12}B_{21} & A_{2} & \cdots & -z_{2N}B_{2N}\\
\hline \vdots & \vdots & \ddots & \vdots\\
\hline -z^*_{1N}B_{N1} & -z^*_{2N}B_{N2} & \cdots & A_{N}
\end{array}\right],
\]
with
\begin{eqnarray*}
A_{i} & = & \one_{2K}\left(\tr X-\tr X_{ii}\right),\ \ \ \ i=1,\dots,N\\
B_{ij} & = & X_{ij}+UX_{ji}^{T}U^{\dg},\ \ \ \ 1\le i<j\le N
\end{eqnarray*}
with $U\in\m{2K}$ denoting an arbitrary unitary
and antisymmetric matrix and $|z_{ij}|\le1$.

\subsection{Positivity}

\begin{prop}
$\Psi$ defines a positive map
\end{prop}

\begin{proof}
The first problem we want to tackle is positivity of a map $\Psi$. Let us consider a rank-1 projector $P=\proj{\psi}{\psi}$, where $\psi\in\mathbb{C}^{N\cdot2K}$ denotes an arbitrary vector. Since $\mathbb{C}^{N\cdot2K}=\bigoplus_{i=1}^{N}\mathbb{C}^{2K}$,
we can represent $\psi$ as follows

\begin{equation}
\ket{\psi}=\bigoplus_{i=1}^{N}\sqrt{\alpha_{i}}\ket{\psi_{i}},
\end{equation}
with $\psi_{i}\in\mathbb{C}^{2K}$, $\alpha_{i}\in[0,1]$ for $i=1,\dots,N$
and $\sum_{i=1}^{N}\alpha_{i}=1$. In addition, for simplicity, we can assume that $\inner{\psi_{i}}{\psi_{i}}=1$ for $i=1,\dots,N$, and then

\begin{equation}
\Psi(P)=\frac{1}{2K(N-1)}\left[\begin{array}{c|c|c|c}
(1-\alpha_{1})\one_{2K} & -z_{12}M_{12} & \cdots & -z_{1N}M_{1N}\\
\hline -z^*_{12}M_{21} & (1-\alpha_{2})\one_{2K} & \cdots & -z_{2N}M_{2N}\\
\hline \vdots & \vdots & \ddots & \vdots\\
\hline -z^*_{1N}M_{N1} & -z^*_{2N}M_{N2} & \cdots & (1-\alpha_{N})\one_{2K}
\end{array}\right],
\end{equation}
with

\begin{equation}
M_{ij}  = \sqrt{\alpha_{i}\alpha_{j}}\Big[\proj{\psi_{i}}{\psi_{j}}+U\proj{\psi^*_{i}}{\psi^*_{j}}U^{\dg}\Big].
\end{equation}
We want to check whether all conditions from Theorem \ref{thm:twierdzenie} are satisfied.  Taking into account that $U$ is an unitary and antisymmetric matrix one has $\inner{\psi_{i}}{U\psi^*_{i}}=\inner{\psi^*_{i}U^{\dg}}{\psi_{i}}=0$, and thus direct calculation leads to

\begin{equation}
M_{ij}M_{ik}^{\dg}  = \alpha_{j}\sqrt{\alpha_{i}\alpha_{k}}\Big[\proj{\psi_{i}}{\psi_{k}}+U\proj{\psi^*_{i}}{\psi^*_{k}}U^{\dg}\Big]=\alpha_{j}M_{ik}.
\end{equation}
By replacing $k$ with $j$ one gets $M_{ij}M_{ij}^{\dg}=\alpha_{j}M_{ii}$. Moreover, since vectors $\ket{\psi_{i}}$ and $U\ket{\psi^*_{i}}$
are mutually orthogonal and normalized, one gets

\begin{equation}
M_{ii}=\alpha_{i}\Big[\proj{\psi_{i}}{\psi_{i}}+U\Big(\proj{\psi_{i}}{\psi_{i}}\Big)^{T}U^{\dg}\Big]\le\alpha_{i}\one_{2K}.
\end{equation}
Therefore we have proved that  $\Psi(P)$ is positive--semidefinite, which completes the proof.
\end{proof}
\subsection{Indecomposibility }

In order to prove that a given map is indecomposable it is enough to find an entangled PPT state $\rho$ such that $\tr(\ew{}\rho)<0$. Let $\ew{\Psi}$ be an EW corresponding to a positive map $\Psi$

\begin{equation}
\ew{\Psi}= \frac 1d \sum_{i,j=1}^{d}e_{ij}\ot W_{ij},
\end{equation}
where $W_{ij}=\Psi(e_{ij})$ and, to simplify notation, we denote $d:=N\cdot 2K$. Let us consider a following construction for the state $\rho$:

\begin{equation}
\rho=\frac{1}{2k+1}\sum_{i,j=1}^{d}e_{ij}\ot\rho_{ij},
\end{equation}
where the diagonal blocks of a state $\rho$ are given by

\begin{equation}
\rho_{ii}=\frac{\one_{d}}{d}-\big(2K(N-1)-1\big)W_{ii}\enskip\text{ for }\enskip i=1,\dots,d,
\end{equation}
and for the off-diagonal blocks one has
\begin{enumerate}
\item if $0<|i-j|<2K$,  then $\rho_{ij}=\mathbb{O}_{d}$,
\item if $|i-j|=2K \ell$, where $\ell=1,\dots,(N-1)$, then for each $\ell$ and $i=1,\dots,2K(N-\ell)$ one has
\begin{equation}
\rho_{i,i+2K \ell}=-W_{i,i+2K\cdot\ell}\ ,
\end{equation}
\item if $|i-j|>2K$, with $|i-j|\neq2K \ell$, then
\begin{equation}
\rho_{ij}=\widetilde{e}_{ij}=\frac{z_{ij}}{(2K)^2 N (N-1)}\, e_{ij}\ ,
\end{equation}
where $\{e_{ij}\}_{i,j=1}^{d}$ stands for an orthonormal basis in
$\m d$.
\end{enumerate}
One proves

\begin{prop}
$\rho \geq 0$ and $\rho^\Gamma \geq 0$, i.e. $\rho$ represents a PPT state.
\end{prop}

\begin{prop}
If  $|z_{ij}|=1$,
then a map $\Psi$ is indecomposible. \end{prop}

\begin{proof}
We will show that $\tr\left(\ew{\Psi}\rho\right)=\sum_{i,j=1}^{d}\tr\left(W_{ij}\rho_{ji}\right)<0$. One has

\begin{equation}
\tr\left(\ew{\Psi}\rho\right)=\frac{1}{\mathcal{N}}\left\{ \sum_{|i-j|=2K\cdot\ell}\tr\left(W_{ij}\rho_{ji}\right)+\sum_{0<|i-j|<2K}\tr\left(W_{ij}\rho_{ji}\right)+\sum_{|i-j|>2K}\tr\left(W_{ij}\rho_{ji}\right)\right\} .\label{eq:slad_wykrywanie_PPT}
\end{equation}
The first sum consists of $N$ terms, the second sum has $2K-1$ terms, and the last one -- $(2K-1)(N-1)$ terms. Straightforward algebra leads to
\begin{eqnarray*}
  \sum_{|i-j|=  2K\cdot\ell}\tr\left(W_{ij}\rho_{ji}\right) &=& 
  \frac{2(K-1)}{(2K)^{3} N(N-1)}\ , \\
\sum_{0<|i-j|<2K}\tr\left( W_{ij}\rho_{ji}\right ) &=&  0, \\
  \sum_{|i-j|>2K}\tr\left(W_{ij}\rho_{ji}\right) &=& 
  - \frac{(2K-1)}{(2K)^{3} N(N-1)}\ ,
\end{eqnarray*}
and, as a consequence,

\begin{equation}
\tr\left(\ew{\Psi}\rho\right)= - \frac{1}{(2K+1)(2K)^{3}N(N-1)}<0,
\end{equation}
which completes to prove.
\end{proof}

\subsection{Optimality}

In a previous section we have shown that if $|z_{ij}|=1$ for all
$1\le i<j<d$, then the map $\Psi$ is non-decomposable. It turns
out that the same condition is necessary and sufficient for optimality.

\begin{prop}\label{prop:optimal}
$\Psi$ is an optimal map if and only if $|z_{ij}|=1$ for all $1\le i<j\le N$.
\end{prop}
\begin{proof}
To show that $|z_{mn}|=1$ is a necessary condition for optimality
we may apply the same argument as in \cite{Justyna3}.
To show that $|z_{mn}|=1$, that is, $z_{mn} = \e^{\im\alpha_{mn}}$, [we introduced $\im$ as an imaginary unit]  is a sufficient condition it is enough
to consider a set of vectors $\Gamma_{\Psi}$ defined as follows

\begin{equation}\label{eq:wektory_do_optymalnosci}
\Gamma_{\Psi}=\{e_{k}\ot e_{k},\ \varphi_{mn}\ot\varphi_{mn},\ \phi_{mn}\ot\widetilde{\phi}_{mn}\ ; \ \text{ for }k=1,\dots,2N\text{ and }1\le m<n\le d\},
\end{equation}
where
\begin{equation}
\varphi_{mn}:=e_{m}+\text{e}^{-\imath\frac{\alpha_{mn}}{2}}e_{n}, \quad\phi_{mn}:=e_{m}+\imath\text{e}^{-\imath\frac{\alpha_{mn}}{2}}e_{n},
\quad\widetilde{\phi}{}_{mn}:=e_{m}-\imath\text{e}^{-\imath\frac{\alpha_{mn}}{2}}e_{n}
\end{equation}
 It is easy to check that elements of a set $\Gamma$ are linearly independent. Direct calculation shows that for each $\psi_{l}\ot\widetilde{\psi}_{l}\in\Gamma_{\Psi}$ the following holds:

\begin{equation}
\langle\psi_{l}\ot\widetilde{\psi}_{l}|\ew{\Psi}|\psi_{l}\ot\widetilde{\psi}_{l}\rangle=0\ ,
\end{equation}
which proves the theorem.  \end{proof}

\subsection{Nd-optimality}

\begin{prop}
$\Psi$ defines a family of nd-optimal maps.
\end{prop}

\begin{proof}
According to Theorem \ref{thm:nd-optimal} to show that $\ew{}$ is nd-optimal it is enough to show, that both $\ew{}$ and $\ew{}^{\Gamma}$ are optimal entanglement witnesses. In Proposition \ref{prop:optimal} we have proved that $\ew{\Psi}$ is optimal. Now we will show that $\ew{\Psi}^{\Gamma}$ is an optimal EW as well. Let us consider the following transformation:

\begin{equation}
(\one_{d}\ot V)\ew{\Psi}(\one_{d}\ot V^{\dg})=\sum_{i,j=1}^{d}e_{ij}\ot V\Psi(e_{ij})V^{\dg},
\end{equation}
where $V:=\one_{N}\ot U^{\dg}$. The action of $\Psi$ on basis elements $\{e_{ij}\}_{i,j=1}^{d}$
is given by

\begin{equation}
\Psi(e_{ij})=\frac{1}{2K\cdot(N-1)}\begin{cases}
(\one_{N}-e_{pp}^{(N)})\ot\one_{2K}\tr\left(e_{rs}^{(2K)}\right), & \mbox{\text{ if }}q=p\\
-z_{ij}\left[e_{pq}^{(N)}\ot e_{rs}^{(2K)}+e_{qp}^{(N)}\ot U\left(e_{rs}^{(2K)}\right)^{T}U^{\dg}\right], & \mbox{\text{ if }}q\neq p
\end{cases}\label{eq:psi_na_bazowych}
\end{equation}
where we introduced  vectors $e_{i},e_{j}\in\C^{d} = \mathbb{C}^N \ot \mathbb{C}^{2K}$ {\em via} the following rule
\[
\begin{cases}
e_{i}= & e_{p}^{(N)}\ot\, e_{r}^{(2K)}\\
e_{j}= & e_{q}^{(N)}\ot\, e_{s}^{(2K)}
\end{cases}\ ,
\]
that is, each $i \in \{1,\ldots,d=N\cdot2K\}$ defines a pair $(p,r)$ with $p\in \{1,\ldots,N\}$ and $r\in \{1,\ldots,2K\}$. It is easy to check, that:

\begin{equation}
V\Psi(e_{ij})V^{\dg}=\frac{1}{2K\cdot(N-1)}(\one_{N}-e_{pp}^{(N)})\otimes\one_{2K}\tr\left(e_{rs}^{(2K)}\right),
\end{equation}
for $q=p$ and

\begin{equation}
V\Psi(e_{ij})V^{\dg} = \frac{-z_{ij}}{2K\cdot(N-1)}\left[e_{qp}^{(N)}\otimes\left(U^{\dg}e_{rs}^{(2K)}U\right)^{T}+e_{pq}^{(N)}\otimes e_{rs}^{(2K)}\right]^{T}
\end{equation}
otherwise. Since $U$ is an unitary and antisymmetric matrix, one has
\begin{equation}
\left(U^{\dg}e_{rs}U\right)^{T}=U^{T}e_{rs}^{T}{U}^*=-Ue_{rs}^{T}{U}^* =Ue_{rs}^{T}U^{\dg}
\end{equation}
 and, as a consequence, one gets $V\Psi(e_{ij})V^{\dg}=\left(\Psi(e_{ij})\right)^{T}$. This shows that

\begin{equation}
(\one_{d}\ot V)\ew{}(\one_{d}\ot V^{\dg})=\ew{}^{\Gamma}.
\end{equation}
Since for all $\psi_{l}\otimes\widetilde{\psi}_{l}\in\Gamma_{\Psi}$, where $\Gamma_{\Psi}$ is defined as (\ref{eq:wektory_do_optymalnosci}), one has

\begin{equation}
\langle\psi_{l}\otimes V^{\dg}\widetilde{\psi}_{l}|\ew{\Psi}^{\Gamma}|V\psi_{l}\otimes\widetilde{\psi}_{l}\rangle=\langle\psi_{l}\otimes\widetilde{\psi}_{l}|\ew{\Psi}|\psi_{l}\otimes\widetilde{\psi}_{l}\rangle=0.
\end{equation}
Direct calculations shows that vectors $\{V\psi_{l}\otimes\widetilde{\psi}_{l};l=1,\dots,d^{2}\}$ are linearly independent. This completes the proof.
\end{proof}

\section{Conclusions}

We provided a new tool which may be used to construct new examples of positive maps (entanglement witnesses) in finite dimensional matrix algebras.  Interestingly, it allows to present a universal proof of positivity  of several well known maps (reduction map, generalized reduction, Robertson map and many others). Finally, it is shown that our method enables one to construct a new family of linear maps and prove that they are positive, indecomposable and even optimal. It should be stressed that this constructions provides linear maps $\Psi : \mathbb{M}_d \rightarrow \mathbb{M}_d$ only for $d = 2N$. It would be interesting to find an analogous method if $d$ is odd. 
For example it would be desirable to provide an appropriate construction generalizing  well known Choi map $\Psi_{\rm Choi} : \mathbb{M}_3 \rightarrow \mathbb{M}_3$ which was proved to be indecomposable and extremal. In a forthcoming paper we plan to report recent progress in this direction.





\bibliographystyle{plain}

\end{document}